\documentclass[tenpt]{IEEEtran}

\usepackage{soul}

\IEEEoverridecommandlockouts
\usepackage{cite}
\usepackage{amsmath,amssymb,amsfonts}
\usepackage{algorithm2e}
\RestyleAlgo{ruled}

\usepackage{graphicx}
\newcommand{\mediumstar}{\mathbin{\scalebox{0.5}{$\bigstar$}}}
\usepackage{subcaption}

\usepackage{textcomp}
\usepackage{xcolor}
\usepackage{amsthm}
\newtheorem{lemma}{Lemma}

\usepackage{mdframed}
\usepackage{enumitem}

\newcommand{\mse}{\mathbb{M}\mathbb{S}\mathbb{E}}
\def\BibTeX{{\rm B\kern-.05em{\sc i\kern-.025em b}\kern-.08em
    T\kern-.1667em\lower.7ex\hbox{E}\kern-.125emX}}
\begin{document}

\title{Decentralized Power Control for Over-the-Air Computation with Phase Noise\\

\thanks{This work was supported in part by ELLIIT, the Swedish Research Council (VR), and the Knut and Alice Wallenberg (KAW) Foundation. Martin Dahl is affiliated with the Wallenberg Autonomous Systems Program (WASP) Graduate School.}
}

\author{\IEEEauthorblockN{Martin Dahl and Erik G. Larsson}\\
\IEEEauthorblockA{\textit{Dept. of Electrical Engineering (ISY)}, \textit{Linköping University, Sweden}\\
martin.dahl@liu.se, erik.g.larsson@liu.se}
}

\maketitle

\begin{abstract}
Estimation of uplink channels is required for coherent over-the-air computation (OAC). When channel estimation is done using calibrated reciprocity, the estimates are only available locally to the devices. This poses a challenge for precoding and decoding, which cannot be coordinated centrally. To this end we use truncated channel inversion (TCI) and propose an approximate closed form solution and an exact numerical solver to optimize the TCI parameters. Importantly, we prove that the proposed TCI scheme is independent of the number of receiver antennas in terms of mean-square-error (MSE). Furthermore, our analysis reveals a clear connection between the MSE and expected aggregate phase error across devices which gives insight to the scalability of OAC. Finally, simulations with comparisons to reference methods from prior work with globally available error-free channel estimates show that proposed is close, even outperforming these references in MSE under some conditions.
\end{abstract}


\section{Introduction} 
Over-the-air computation (OAC) is a  method to compute the sum of values from distributed devices at a central unit, such as an access point (AP). This aggregation is useful, for example in sensor networks and for federated learning (FL), where the sum rather than individual values from devices is of interest \cite{zhu2019broadband}. With traditional digital communication, each device must transmit their values over distinct wireless resources, such as bandwidth, time and space, whereafter the central unit computes the sum. This leads to a resource use growing linearly with the number of devices. Instead, OAC exploits the superposition principle of wave propagation as transmitted values add up at the receiver, giving both communication and computation over a single resource. In principle, OAC makes latency and bandwidth independent of the number of devices. 

However, in practice the devices must precode their transmitted values with respect to their uplink channels in order for their values to add up coherently in phase at the receiver. For this coherent precoding, the instantaneous uplink channels must be estimated, which creates an overhead that scales linearly with the number of devices. For example, each device may transmit an uplink pilot to which the AP responds with the channel estimate. To reduce this overhead, non-coherent OAC has been proposed \cite{goldenbaum2013robust, csahin2023over}, relying on amplitude modulation which preserves information even without phase-coherency. Then channel estimation is not necessary. However, such non-coherent aggregation requires repeated transmissions to reduce the uncertainty from the randomly rotated transmit phases, resulting in poor energy and resource efficiency. Alternatively, for coherent OAC one may exploit channel reciprocity, which enables uplink channel estimation using a single downlink pilot from the AP. 

Unfortunately, reciprocity does not hold in practice because of transmitter and receiver hardware, unless the devices calibrate for joint reciprocity through bidirectional signaling. This signaling can be with the AP centrally or with the other devices in a decentralized manner \cite{nissel2022correctly, shepard2012argos}. Even then, all transceiver oscillators are impacted by phase noise, which can cause a drift in phase depending on whether the oscillators are free-running or locked \cite{piemontese2024discrete}, meaning calibration has to be done repeatedly which again incurs overhead. If the oscillators are sufficiently stable, then calibration can be done seldom, enabling an efficient way of achieving reciprocity and uplink channel estimates \cite{dahl2024over}. Instead of calibration, one may continuously broadcast a signal over-the-air as a common phase reference \cite{abari2015airshare}, however, this is more costly in terms of energy and wireless resources than periodic calibrations. 

When reciprocity-based channel estimation is used (with appropriate calibration), the uplink channels are only known at the devices in a decentralized fashion, not at the AP centrally. This poses an issue for optimal precoding/decoding as centralized channel knowledge is required to compensate for devices with instantaneously weak channels, which is often assumed in power control schemes for OAC \cite{cao2020optimized, xie2023optimal, chen2018uniform, tang2023miso, zhu2018mimo}. A promising decentralized power control scheme is truncated channel inversion (TCI) \cite{zhu2019broadband, cao2020optimized}, \cite{mital2022bandwidth, abrar2025biased, amiri2020federated}. However, most of these papers are in the context of FL without focus on OAC itself, where the former considers optimizing a convergence rate, the latter an estimation error. There are also schemes where the devices have no estimates but the AP estimates the channel sum, which is cheap to acquire \cite{zhai2021hybrid, becirovic2022optimal, jing2023transceiver}: these perform poorly unless the number of AP antennas is larger than the number of devices, which is demonstrated herein.  

\textbf{Contribution: }This paper aims to answer what the optimal precoder and decoder is for OAC when instantaneous channel estimates are only available at the devices locally. Specifically, we look closer at TCI  in the context of OAC with reciprocity-based channel estimation and phase noise. To compute the optimal TCI parameters we derive an approximate closed-form solution, which to the best of our knowledge has not been done previously, as well as a numerical solver. We also prove that the performance of the optimized TCI scheme is independent of the number of receiver antennas. This independence is noteworthy as non-coherent OAC, with even less channel information, does benefit from multiple antennas \cite{csahin2023over}. The proposed methods are simulated and compared to previously published methods with globally available error-free channel estimates. 

\section{System Model and Problem Formulation}
Let $K$ single-antenna devices hold data $x_k\in\mathbb{C}$ where $\mathbb{E}[|x_k|^2]=1$, $\mathbb{E}[x_k]=0$. The goal is for a central AP with $M$ antennas to acquire an estimate $\widehat{x}$ of the sum 
\begin{equation}
    x \equiv \sum_{k=1}^{K}x_k,
\end{equation}
minimizing the mean-square error (MSE) $\mse(x,\widehat{x})\equiv\mathbb{E}\left[|x-\widehat{x}|^2\right]$. To this end, each device transmits a value $b_kx_k$ where $b_k\in\mathbb{C}$ is a precoding coefficient independent of $x_k$ with power constraint $\mathbb{E}[|b_k|^2]\leq P$. 
Then, the AP receives
\begin{equation}
    \mathbf{y} = \sum_{k=1}^{K}\mathbf{h}_kb_kx_k + \mathbf{n} \in \mathbb{C}^M,
\end{equation}
where $\mathbf{h}_k\in\mathbb{C}^M$ is the channel coefficient, and  $\mathbf{n}\in\mathbb{C}^M$ zero-mean noise with identity covariance. To estimate $x$, the AP applies a linear combiner $\mathbf{a}\in\mathbb{C}^M$:
\begin{equation}
    \widehat{x} \equiv \mathbf{a}^\text{H}\mathbf{y}.
\end{equation}
The minimum MSE problem thus follows as
\begin{equation}
\begin{split}
    &\underset{\mathbf{a}, \{b_k\}}{\text{min }}\mse(x,\widehat{x})\\
    &\text{s.t. }\mathbb{E}[|b_k|^2]\leq P,
\end{split}
\end{equation}
where $\mse(x,\widehat{x})$, conditional on $\mathbf{h}_k, \mathbf{a}$ and $b_k$ is
\begin{equation}\label{eq:MSE1}
\begin{split}
    \mathbb{E}\left[|x-\widehat{x}|^2\big|\{\mathbf{h}_k, b_k\}, \mathbf{a}\right]= \sum_{k=1}^{K}|1-\mathbf{a}^\text{H}\mathbf{h}_kb_k|^2 + \mathbf{a}^\text{H}\mathbf{a}. 
\end{split}
\end{equation}

\section{Proposed Scheme}
First the devices calibrate for reciprocity, then rely on downlink pilots from the AP to estimate the channel to the AP. This means channel estimates are only available locally to the devices, and the optimal precoder and decoder must be computed in isolation by each device and the AP.
\subsection{Channel Estimation with Reciprocity Calibration}
For simplicity, consider a single antenna of the AP, say $m$. The \textit{effective} channel from device $k$ to antenna $m$ ($h_{k\rightarrow m}\equiv [\mathbf{h}_k]_m$)\footnote{$[\cdot]_m$ denotes element $m$.} is not reciprocal to the corresponding downlink $h_{m\rightarrow k}$ because of transceiver hardware. This is modeled as  $h_{k\rightarrow m} = t_kg_{k,m}r_m$ and $h_{m\rightarrow k} = t_mg_{k,m}r_k$ where $g_{k,m}\in\mathbb{C}$ is the reciprocal channel between the device and AP antenna and $t_k, r_k, t_m, r_m\in\mathbb{C}$ the transmitter and receiver gains \cite{nissel2022correctly}. To calibrate, each device $k$ transmits a pilot to the AP which measures $h_{k\rightarrow m}$ (ignoring noise). The AP responds with a precoded pilot to each device, receiving
\begin{equation}\label{eq:calibration_coeff_meas}
	y_k = \frac{h_{m\rightarrow k}}{h_{k\rightarrow m}} = \frac{r_k}{t_k}\frac{t_m}{r_m}\equiv \frac{1}{c_k}.
\end{equation}
Then, with a single broadcast pilot from the AP over a new channel instance, say $g_{k,m}'\neq g_{k,m}$, the devices estimate the new uplink $h_{k\rightarrow m}'$ by multiplying the downlink with the calibration coefficient $c_k$: $c_kh_{m\rightarrow k}' = g_{k,m}'t_k r_m = h_{k\rightarrow m}'$. The AP must apply the same procedure as above internally, with antenna $m$ as reference to calibrate its $M$ antennas\cite{shepard2012argos}. In turn, the devices only calibrate to antenna $m$. 


\subsection{Phase Drift from Wiener Phase Noise}
Relative to the amplitudes $|t_k|, |r_k|, \forall k$ and the AP hardware $t_m, r_m$ of antenna $m$, the device hardware phases $\angle t_k, \angle r_k$ will drift significantly faster with time \cite{nissel2022correctly}. This drift is primarily caused by phase noise, which for free-running oscillators can be modeled as a Wiener process \cite{piemontese2024discrete} of Gaussian increments $\phi$ to the transmit and receive chain sharing the same oscillator \cite{nissel2022correctly}: $\angle t_k' = \angle t_k + \phi_k/2$,  $\angle r_k' = \angle r_k - \phi_k/2$, $\phi_k\sim\mathcal{N}(0,\beta)$, resulting in the estimate 
\begin{equation}\label{eq:reciprocity_channel_estimation_phase_error}
	\widehat{h_{k\rightarrow m}'} = h_{k\rightarrow m}'e^{j\phi_k}.
\end{equation}
After a time $T$ since calibration the phase noise will accumulate to $\phi_k\sim\mathcal{N}(0,T\beta)$ such that (\ref{eq:reciprocity_channel_estimation_phase_error}) becomes increasingly uncertain with time. The phase noise variance $\beta\in\mathbb{R}$ corresponds to the quality of an oscillator where lower $\beta$ is better.

\subsection{Proposed Precoder and Decoder}
Given a decoder $\mathbf{a}$, the MSE (\ref{eq:MSE1}) is minimized by 
\begin{equation}
    b_k^{\mediumstar} = \frac{\mathbf{h}_k^\text{H}\mathbf{a}}{|\mathbf{h}_k^\text{H}\mathbf{a}|^2},
\end{equation}
however, determining the optimal $\mathbf{a}$ under a power constraint requires global instantaneous knowledge of $\mathbf{h}_k$. To handle the power constraint our proposed decentralized scheme is an heuristic combination of $b_k^{\mediumstar}$ and TCI:
\begin{equation}\label{eq:proposed_scheme}
\begin{split}
    &\mathbf{a} = \sqrt{\eta}\mathbf{1},\\
    &b_k = \begin{cases}
            \frac{\widehat{\mathbf{h}_k}^\text{H}\mathbf{a}}{\big|\widehat{\mathbf{h}_k}^\text{H}\mathbf{a}\big|^2} & \text{if }\big|\widehat{\mathbf{h}_k}^\text{H}\mathbf{a}\big|\geq\gamma,\\
            0 & \text{else},
            \end{cases}
\end{split}
\end{equation}
where $\widehat{\mathbf{h}_k}$ is an estimate of $\mathbf{h}_k$ available to device $k$ only, $\gamma$ and $\eta$ are positive constants. The estimate is then
\begin{equation}
    \widehat{x} = \sum_{k=1}^{K}\mathbf{1}^T\mathbf{h}_k\frac{\widehat{\mathbf{h}_k}^\text{H}\mathbf{1}}{\big|\widehat{\mathbf{h}_k}^\text{H}\mathbf{1}\big|^2}x_kI_k + \sqrt{\eta}\mathbf{1}^T\mathbf{n},
\end{equation}
where $I_k=1$ if $\big|\widehat{\mathbf{h}_k}^\text{H}\mathbf{a}\big|\geq\gamma$, else $0$. The proposed TCI scheme is the first to consider multiple antennas at the receiver, but is otherwise similar to previous work.

\subsection{Minimum MSE Optimization Objective}
Let $\mathbf{h}_k\sim\mathcal{CN}(0,\sigma^2\mathbf{I})$, assuming homogeneous channel powers\footnote{This is commonly assumed by previous work and simplifies the analysis. The heterogeneous case will primarily differ in that devices should have different truncation thresholds, increasing the optimization complexity.} giving a transmission probability of $P\left(I_k=1\right)  = e^{-\frac{\gamma^2}{\eta M\sigma^2}}$. Furthermore, let $\widehat{\mathbf{h}_k} = \mathbf{h}_ke^{j\phi_k}$ for some phase error $\phi_k\in\mathbb{R}$. Assume $\phi_k, \forall k$ are independent and let $\alpha_k\equiv\mathbb{E}[e^{j\phi_k}]$ such that $\alpha_k\in\mathbb{R}$. Then $\widehat{x}$ is given as follows
\begin{equation}
    \widehat{x} = \sum_{k=1}^{K}x_ke^{-j\phi_k}I_k + \sqrt{\eta}\mathbf{1}^T\mathbf{n},
\end{equation}
with the $\mse(x,\widehat{x})$ as a function of $\eta$ and $\gamma$
\begin{equation}\label{eq:MSE3}
\begin{split}
    \mse(\eta, \gamma) \equiv K + e^{-\frac{\gamma^2}{\eta M\sigma^2}}\Sigma + \eta M,
\end{split}
\end{equation}
where $\Sigma \equiv \sum_{k=1}^{K}\left(1-2\alpha_k\right)$. Here $\Sigma$ is a quantity describing the expected aggregate phase error which should be as low as possible. We want to optimize $\eta, \gamma$ by solving

\begin{equation}\label{eq:MSE_objective_1}
\begin{split}
    &\eta^{\mediumstar}, \gamma^{\mediumstar} =\underset{\eta, \gamma}{\text{ arg min }}\mse(\eta, \gamma)\\
    &\text{s.t. }\mathbb{E}[|b_k|^2]\leq P,\gamma,\eta > 0.
\end{split}
\end{equation}

\section{Minimizing the MSE}
First consider the case of $\Sigma\geq 0$. Then $\mse(\eta, \gamma)> K$ which is worse than $\widehat{x}=0$ with $\mse(x,\widehat{x})=K$ and the proposed TCI scheme should not be used. This makes $\Sigma=0$ a useful threshold to determine when devices are too out of phase in expectation, which should trigger a re-calibration \cite{dahl2024over}. Now consider $\Sigma<0$ and $\mathbf{h}_k^\text{H}\mathbf{1}\equiv\widetilde{h}_k\sqrt{\frac{M\sigma^2}{2}}$ where $\widetilde{h}_k\sim\mathcal{CN}(0,2)$ such that $|\widetilde{h}_k|^2\sim\text{exp}(1/2)$, then

the exact second moment of $b_k$ is
\begin{equation}\label{eq:closed_form_constraint}
\begin{split}
    &\mathbb{E}[|b_k|^2] = \int_{\frac{2\gamma^2}{\eta M\sigma^2}}^{\infty}\frac{2}{\eta M\sigma^2x}\frac{1}{2}e^{-\frac{x}{2}}dx=\frac{\textit{E}_1\left(\frac{\gamma^2}{\eta M\sigma^2}\right)}{\eta M\sigma^2},
\end{split}
\end{equation}
where $E_1$ is the \textit{exponential integral} \cite{abramowitz1965handbook}. 
\begin{lemma}\label{lemma_etaM}
    With the proposed TCI scheme (\ref{eq:proposed_scheme}), the minimum $\mse(\eta,\gamma)$ in (\ref{eq:MSE_objective_1}) using $\eta^{\mediumstar}$ and $\gamma^{\mediumstar}$ is independent of $M$.
\end{lemma}
\begin{proof}
    Note that $\eta$ and $M$ always appear as $\eta M$ in the $\mse(\eta, \gamma)$ and the closed form constraint (\ref{eq:closed_form_constraint}). Since both $\eta$, $M$ are positive, we can substitute $\eta M$ with $\mu>0$ and instead minimize $\mse(\eta, \gamma)$ w.r.t. $\gamma$, $\mu$, independent of $M$. 
\end{proof}
As a consequence of Lemma \ref{lemma_etaM} and the distribution of $\mathbf{h}$ being statistically white, the MSE only depends on  $||\mathbf{a}||$. For example $\mathbf{a}=\sqrt{\eta}[1,0,\dots,0]^\text{T}$ yields equal performance as $\mathbf{a}=\sqrt{\eta}\mathbf{1}$. Note that calibration of active AP antennas (non-zero elements of $\mathbf{a}$) is still necessary.  
\subsection{Approximate Closed-Form Solution}
For $x>0$ we have \cite{abramowitz1965handbook} 
\begin{equation}
    \textit{E}_1(x) \leq e^{-x}\text{ln}\left(1 + \frac{1}{x}\right)\leq \frac{e^{-x}}{\sqrt{x}},
\end{equation}
where the final inequality follows from $\text{ln}(1+x)\leq\sqrt{x}$ for $x\geq 0$.  Inserted into (\ref{eq:closed_form_constraint}) this inequality yields
\begin{equation}\label{eq:power_upper_bound}
\begin{split}
    &\mathbb{E}[|b_k|^2] \leq \frac{e^{-\frac{\gamma^2}{\eta M\sigma^2}}}{\eta M\sigma^2\sqrt{\frac{\gamma^2}{\eta M\sigma^2}}} = \frac{e^{-\frac{\gamma^2}{\eta M\sigma^2}}}{\sqrt{\sigma^2\gamma^2\eta M}}.
\end{split}
\end{equation}
By enforcing maximum power at the upper bound
\begin{equation}\label{eq:approx_power_constraint}
\begin{split}
    \frac{e^{-\frac{\gamma^2}{\eta M\sigma^2}}}{\sqrt{\sigma^2\gamma^2\eta M}} = P
    \Leftrightarrow
    e^{-\frac{\gamma^2}{\eta M\sigma^2}} = P\sqrt{\sigma^2\gamma^2\eta M},
\end{split}
\end{equation}
a solution to $\gamma^2$ is given by
\begin{equation}\label{eq:exact_gamma2}
    \gamma^2 = \frac{\eta M\sigma^2}{2}\text{W}\left(\frac{2}{P^2\eta^2\sigma^4 M^2}\right)\equiv \widetilde{\gamma}^2,
\end{equation}
where W is the principal branch of the \textit{Lambert W} function\footnote{$\text{W}(xe^x)=x$, $x>-1/e$.} \cite{corless1996lambert}. Note that if $\eta$ is fixed, then $\gamma^{\mediumstar}\leq\widetilde{\gamma}$ for that $\eta$ since (\ref{eq:closed_form_constraint}) and (\ref{eq:power_upper_bound}) decrease monotonously with increasing $\gamma>0$. As the argument of $\text{W}(\cdot)$ in (\ref{eq:exact_gamma2}) is positive we have\footnote{Is easiest shown by plotting $\text{W}(\cdot)$ against the upper bound.} 
\begin{equation}\label{eq:gamma2_upperbound}
\begin{split}
    &\gamma^2 \leq \frac{\eta M\sigma^2}{2}\text{ln}\left(1+\frac{2}{P^2\eta^2\sigma^4 M^2}\right).
\end{split}
\end{equation}
 If (\ref{eq:approx_power_constraint}) is enforced, then (\ref{eq:MSE3}) and (\ref{eq:gamma2_upperbound}) yields 
\begin{equation}
\begin{split}
    &\mse(\eta, \gamma) = K + \Sigma P\sqrt{\sigma^2\gamma^2\eta M} + \eta M\\
    &\leq K + \Sigma \frac{P\sigma^2\eta M}{\sqrt{2}}\sqrt{\text{ln}\left(1+ \frac{2}{P^2\eta^2\sigma^4 M^2}\right)} + \eta M\\
    &\approx K + \Sigma \frac{P\sigma^2\eta M}{\sqrt{2}}\text{ln}\left(1+ \frac{\sqrt{2}}{P\eta\sigma^2 M}\right) + \eta M,
\end{split}
\end{equation}
which for $\Sigma<0$ is minimized by
\begin{equation}\label{eq:approx_eta}
\begin{split}
    &\eta = \frac{\sqrt{2}}{P\sigma^2 M}\left(\frac{1}{1 + \text{W}\left(-e^{\frac{\sqrt{2}}{\Sigma P\sigma^2} - 1}\right)} - 1\right)\equiv\widetilde{\eta}.
\end{split}
\end{equation}
Note that $\text{W}\left(-e^{\frac{\sqrt{2}}{\Sigma P\sigma^2} - 1}\right)$ is negative and greater than $-1$, so $\widetilde{\eta}$ is positive. Finally, the closed-form solution follows below:
\begin{mdframed}
\textbf{Approximate Closed-Form Solution}
\begin{equation*}
    \begin{split}
        &\eta^{\mediumstar} \overset{(\ref{eq:approx_eta})}{=} \widetilde{\eta}, \text{ } \gamma^{\mediumstar} = \frac{\widetilde{\eta} M\sigma^2}{2}\text{W}\left(\frac{2}{P^2\widetilde{\eta}^2\sigma^4 M^2}\right).
    \end{split}
\end{equation*}
\end{mdframed}
\subsection{Exact Grid-Search Solution}
To solve (\ref{eq:MSE_objective_1}), $\eta^{\mediumstar}, \gamma^{\mediumstar}$ are found through grid- and bisection-search, guided by the upper bounds in Lemma \ref{lemma_bound}:
\begin{lemma}\label{lemma_bound}
    \begin{equation}
    \begin{split}
       & \eta^{\mediumstar} \leq \frac{L}{M},\text{ }\gamma^{\mediumstar} \leq \frac{1}{\sqrt{P\sqrt{2}}},
    \end{split}
    \end{equation}
    where $L\equiv\mse(\widetilde{\eta}, \widetilde{\gamma})$.
\end{lemma}
\begin{proof}
    Using $\widetilde{\eta}$ and $\widetilde{\gamma}$ satisfies $\mathbb{E}[|b_k|^2]\leq P$, which with $\widetilde{\gamma}$ follows from (\ref{eq:power_upper_bound})-(\ref{eq:exact_gamma2}) for any $\eta$. Then $L$ is an upper bound to $\mse(\eta^{\mediumstar}, \gamma^{\mediumstar})$. Furthermore, from $(\ref{eq:MSE3})$ we get that $\mse(\eta, \gamma)\geq\eta M$ since $e^{-\frac{\gamma^2}{\eta M\sigma^2}}\Sigma\in[-K,0]$. Then $\eta^{\mediumstar} M \leq L$ which proves the first bound. The second bound can be shown using (\ref{eq:gamma2_upperbound}) with $\text{ln}(1+x)\leq \sqrt{x}$ for $x>0$.
\end{proof}
The grid-bisection-search is performed as follows, where $N_\eta$ is the number of $\eta$ values and $\eta_\text{min}$ is the smallest value in the $\eta$-grid:
\begin{samepage}\begin{mdframed}
\textbf{Grid-Bisection-Search (GS) Solution}
\begin{enumerate}[label=\arabic*., leftmargin=*, itemsep=0.3em]
  \item Let $\mathcal{N}$ be a range of $N_\eta$ $\eta$ values evenly spaced from $\eta_\text{min}$ to $\frac{L}{M}$.
  \item For each $\eta\in\mathcal{N}$, find $\gamma\equiv\gamma_\eta^{\mediumstar}$ such that \\$|\mathbb{E}[|b_k|^2]-P|<\epsilon$ for $\eta>0$, using bisection-search starting at $\gamma\in\{0,\widetilde{\gamma}\}$, where $\widetilde{\gamma}$ is computed using (\ref{eq:exact_gamma2}). The power $\mathbb{E}[|b_k|^2]$ is computed using (\ref{eq:closed_form_constraint}). The number of bisection iterations is capped to 200.
  \item Out of all pairs $(\eta\in\mathcal{N}, \gamma^{\mediumstar}_\eta)$ select the pair minimizing (\ref{eq:MSE3}) as $\eta^{\mediumstar}, \gamma^{\mediumstar}$.
\end{enumerate}
\end{mdframed}
\end{samepage}
While the GS solution is more accurate than the approximate closed-form solution, the closed-form solution is more interpretable, enabling a direct analytical study of the TCI scheme. If system parameters change rapidly, the decreased computational complexity of the closed-form solution is useful. 

\section{Statistical Phase-Error limit}
Now consider the behavior of $\Sigma$ under the effect of Wiener phase noise. Furthermore, consider $T$ sequential rounds of OAC indexed by $t$ with data $x_k^t$. Assume $\phi_k\sim\mathcal{N}(0,\Phi_k)$ for some $\Phi_k>0$, then $e^{j\phi_k}$ is log-normal such that $\alpha_k = e^{-\frac{1}{2}\Phi_k}$ \cite{johnson1995continuous}. Let devices calibrate sequentially on the same sub-carrier starting at $t=-K+1$. Then $\Phi_k$ is time-dependent as $\Phi_k^t = (t+k - 1)\beta$ and $\Sigma$ is a geometric sum
which is negative if
\begin{equation}\label{eq:phase_noise_inequality}
t <-\frac{2}{\beta}\text{ln}\left(\frac{K}{2}\left(\frac{1-e^{\frac{-\beta}{2}}}{1 - e^{\frac{-\beta K}{2}}}\right)\right)\equiv T.
\end{equation}
That is, phase noise limits either the number of OAC rounds $T$, or how many devices can participate for any $T$. Additionally, the maximum $\beta$ can be found for any $K$ by solving for $T=0$, which gives insight to how stable oscillators must be for OAC. This is demonstrated by Fig. \ref{fig:sigma_boundary}. Finally, the case of $K=1$ can be interpreted as having each device calibrate on separate sub-carriers simultaneously. While sequential calibration induces a resource growth with $K$, if the oscillators are stable (low $\beta$), calibration can be done seldom (high $T$) compared to feedback-based channel estimation in every coherence block.

\begin{figure}[t!]
    \centering
    \begin{subfigure}{\linewidth}
        \centering
        \includegraphics[width=0.8\linewidth]{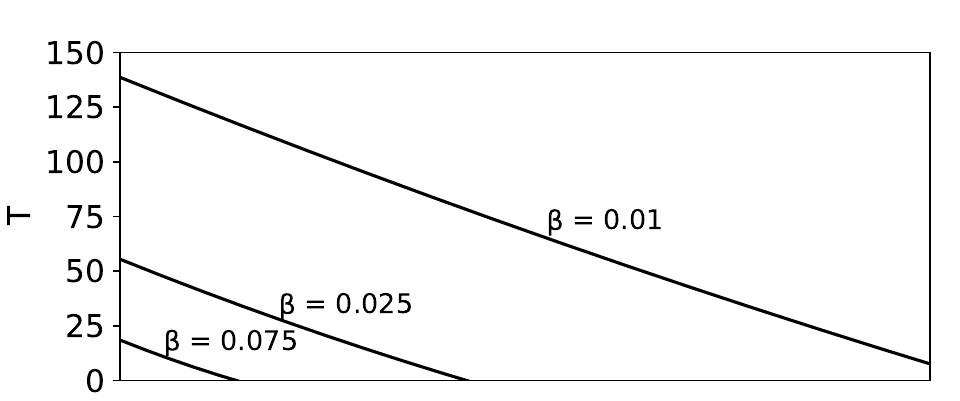}
    \end{subfigure}
    \begin{subfigure}{\linewidth}
        \centering
        \includegraphics[width=0.8\linewidth]{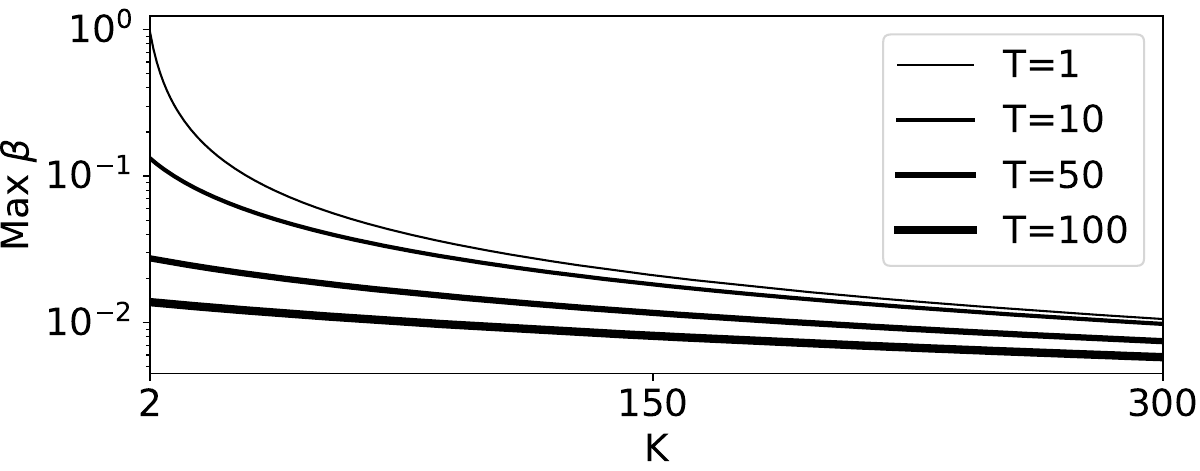}
    \end{subfigure}
    \caption{The upper figure shows $T$ against $K$ for various $\beta$. The lower figure shows the maximum $\beta$ such that $t$ fulfills (\ref{eq:phase_noise_inequality}) up to the corresponding $T$ specified by the legend.}
    \label{fig:sigma_boundary}
\end{figure}

\section{Simulations}
Through simulation, the proposed TCI scheme with solutions are compared to reference methods from prior work with error-free channel estimates globally available to all devices and the AP:
\begin{enumerate}
    \item \textbf{SUM}: $\mathbf{a}$ is selected as equation (36) from \cite{zhai2021hybrid}, a sum-combiner, while $b_k$ is selected as $\sqrt{P}\frac{(\mathbf{1}^T\mathbf{h}_k)^*}{|\mathbf{1}^T\mathbf{h}_k|}$.
    \item \textbf{PGDM}: Corresponds to the ``PGDM" method from \cite{jing2023transceiver} with step-size $1$, momentum $1$ and $100$ steps.
\end{enumerate}
  For all settings, the MSEs of these references are simulated using $10^5$ Monte Carlo trials, where the channel, noise ($\mathbf{n}\sim\mathcal{CN}(\mathbf{0},\mathbf{I})$) and data ($x_k\sim\mathcal{CN}(0,1)$, i.i.d. $\forall k$) is drawn in every trial. Our proposed ``TCI Approx" corresponds to using the approximate closed-form solution, while ``TCI GS" corresponds to using the grid-bisection-search, where the MSE is computed using (\ref{eq:MSE3}). For TCI GS we use $N_\eta=2000, \eta_\text{min}=10^{-5}$, $\epsilon=10^{-5}$.

In Fig. \ref{fig:MSE_antennas} the references are compared to proposed for various number of receiver antennas. These results apply for $\Sigma=-K$, corresponding to devices being perfectly aligned in phase, which is a best-case scenario. Firstly, the MSE is indeed independent of $M$ for the proposed method. Secondly, there is a drop in MSE when $M\geq K$ for the references, demonstrating that $M$ needs to be sufficiently large.  This drop is most visible with a large $P$, and in that case proposed even outperforms the references for $M<K$. This highlights the beneficial scaling with $K$ using the TCI scheme, compared to the reference methods. While the MSE is relatively high ($\sim K/10$) for low $P$, it decreases fast with $P$ and $\sigma^2$:

In Fig. \ref{fig:MSE_power} and \ref{fig:MSE_sigma2} the MSE is simulated and computed for various $P$ and $\sigma^2$, respectively. For large $P$ one observes that PGDM fails, which is reasonably due to numerical instability and difficulties tuning the PGDM hyperparameters. Otherwise, the TCI GS appears to perform similarly to the references for large $P$ and $\sigma^2$.

Fig. \ref{fig:MSE_SIGMA} shows how a larger $\Sigma$ increases the MSE. Importantly, for $P=10^5$ one gets a best-case scenario, which appears to be a linear MSE growth with $\Sigma$ starting at $\mse(\eta,\gamma)=0$. 

\begin{figure}[t!]
    \centering
    \begin{subfigure}{\linewidth}
        \centering
        \includegraphics[width=0.8\linewidth]{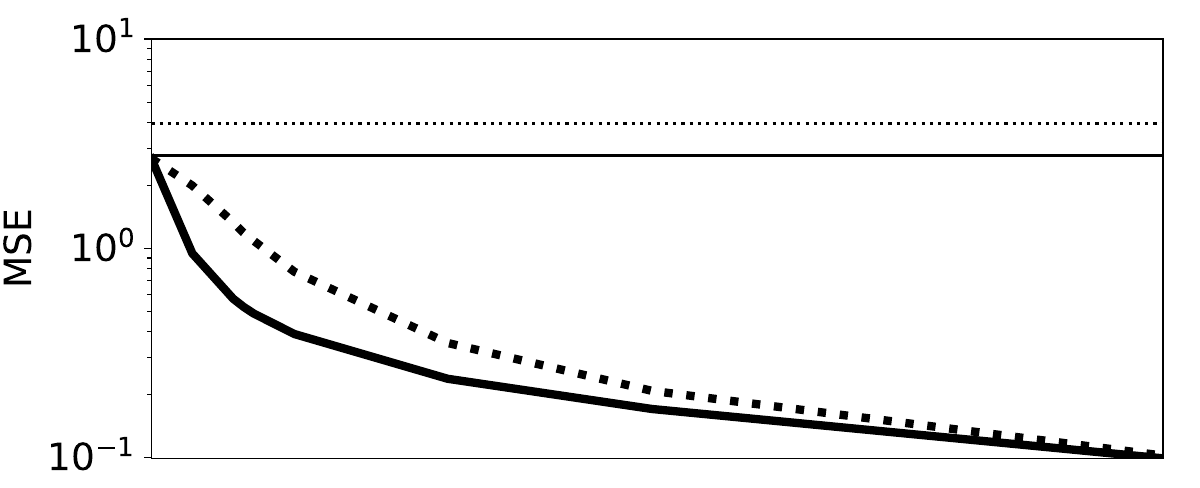}
    \end{subfigure}
    \begin{subfigure}{\linewidth}
        \centering
        \includegraphics[width=0.8\linewidth]{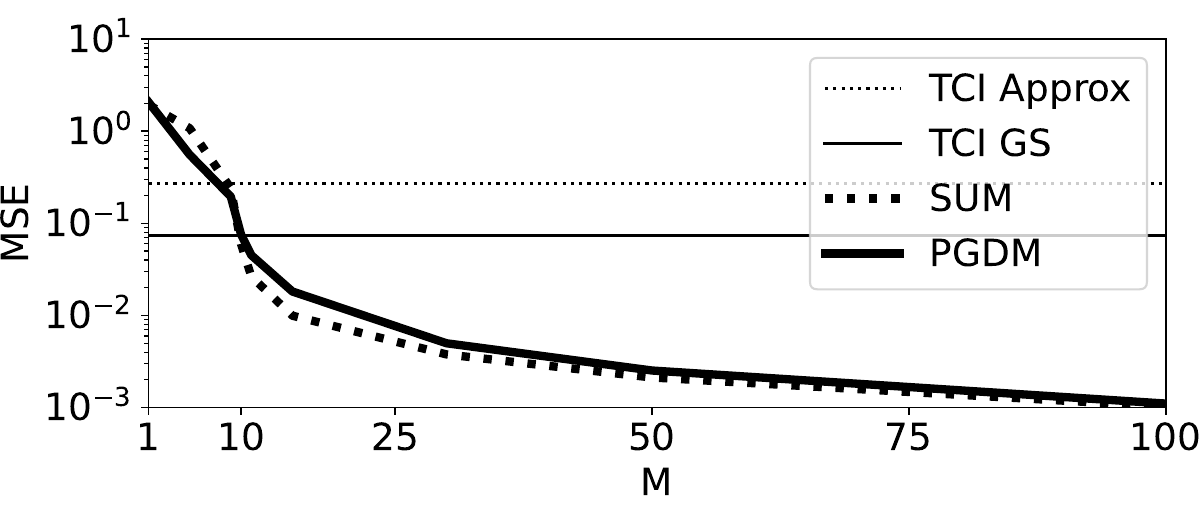}
    \end{subfigure}
    \caption{$\mse(\eta, \gamma)$ vs $M$. The upper figure has $P=1$, the lower $P=100$. $K=10, \sigma^2=1, \Sigma=-K$.}
    \label{fig:MSE_antennas}
\end{figure}

\begin{figure}[t!]
  \centering
  \includegraphics[width=0.85\linewidth]{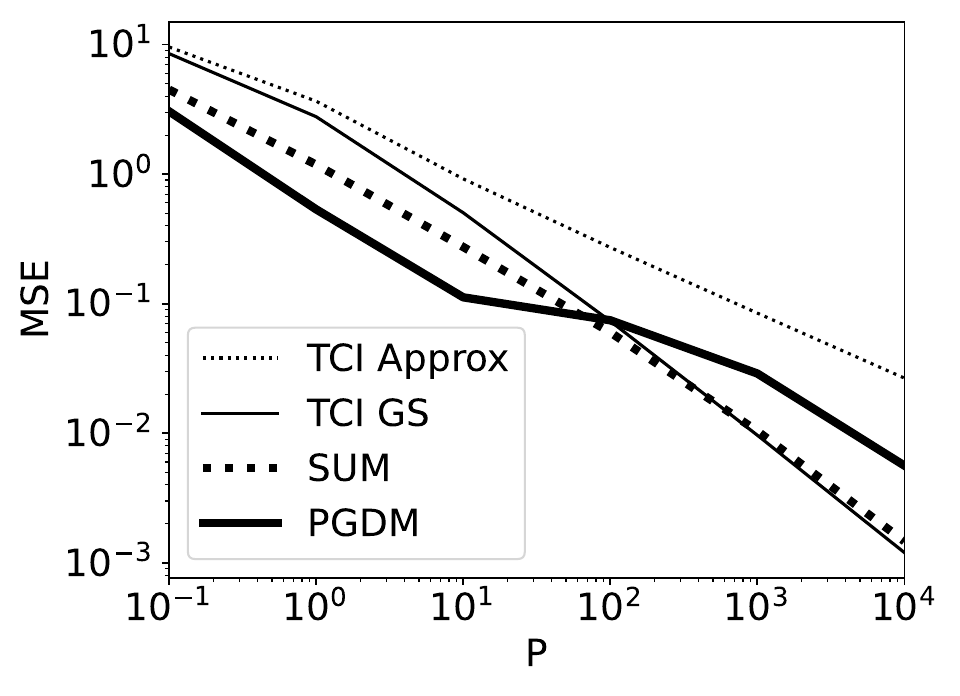}
  \caption{$\mse(\eta, \gamma)$ vs $P$. $M=10, K=10, \sigma^2=1, \Sigma=-K$.}
  \label{fig:MSE_power}
\end{figure}
\begin{figure}[t!]
  \centering
  \includegraphics[width=0.85\linewidth]{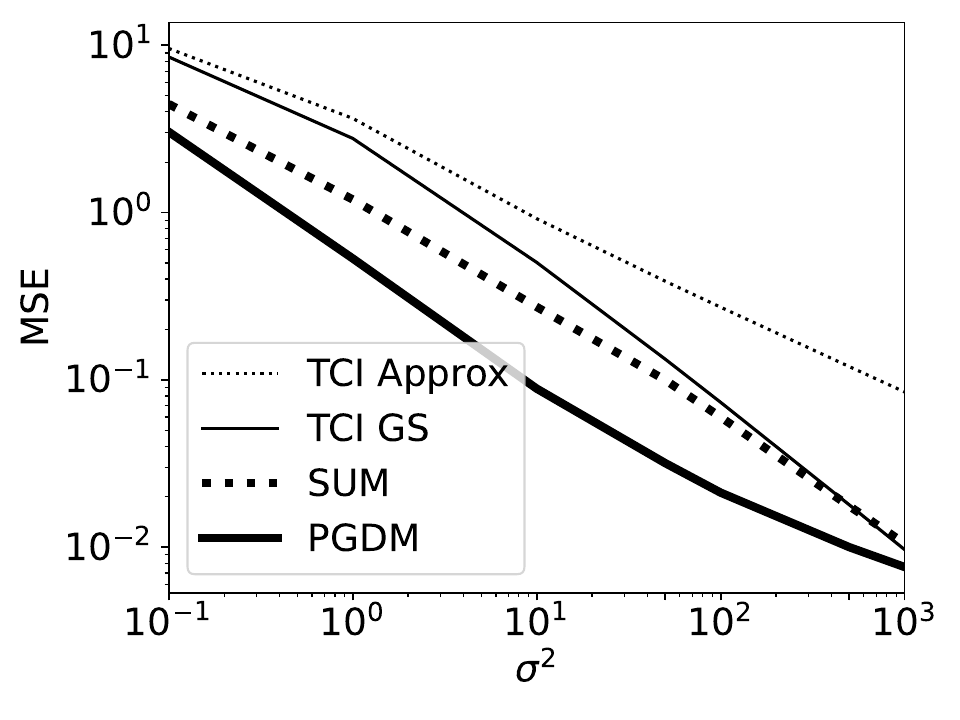}
  \caption{$\mse(\eta, \gamma)$ vs $\sigma^2$. $M=10, P=1, K=10, \Sigma=-K$.}
  \label{fig:MSE_sigma2}
\end{figure}
\begin{figure}[t!]
  \centering
  \includegraphics[width=0.85\linewidth]{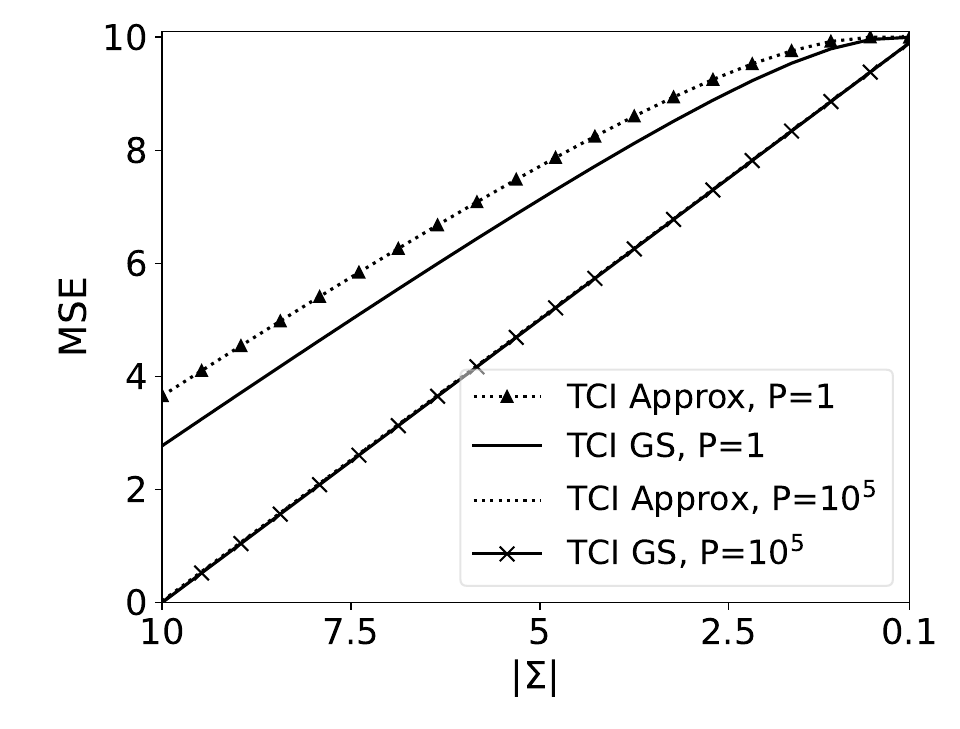}
  \caption{$\mse(\eta, \gamma)$ vs $\Sigma<0$. $M=1, \sigma^2=1, K=10$.}
  \label{fig:MSE_SIGMA}
\end{figure}

\section{Conclusion}
A TCI scheme for OAC was investigated, where channel estimates are only available at transmitters. We derived a closed-form and numerical solver to optimize our TCI scheme, which demonstrated competitive performance to references with globally available error-free channel estimates. We showed that increasing the number of antennas gives no benefit to the proposed TCI scheme at its optima. In an extended version we will look for a precoder and decoder where the MSE is improved by an increased number of antennas, for example by using non-linear processing. Furthermore, we plan to look at the analytical performance of the closed-form solution and generalize the system model to heterogeneous channel powers. 

\bibliographystyle{IEEEtran}
\bibliography{bib}
\end{document}